\documentclass[final]{l4dc2025}

\usepackage{caption}
\usepackage{subcaption}
\usepackage{times}
\setcitestyle{authoryear,open={(},close={)}} 
\usepackage{tikz}
\usetikzlibrary{calc,shapes,arrows}
\usepackage{xspace}
\usepackage{algorithm}
\usepackage{algorithmic}
\usepackage{multicol}
\usepackage{nccmath,amsmath}

\newcommand{\R}{{\mathbb{R}}}

\newcommand{\f}{\mathcal{F}}

\newcommand{\N}{{\mathbb{N}}}

\newcommand{\st}{{\it s.t. \space}} 
\newcommand{\argmin}{\mathrm{arg}\min} 

\newcommand{\ra}{\rightarrow}

\newcommand{\transpose}[1]{\ensuremath{#1^{ \top}}}

\newcommand{\Ad}{\hat{A}}
\makeatletter
\newcommand*{\rom}[1]{\expandafter\@slowromancap\romannumeral #1@}
\makeatother
\newcommand{\overbar}[1]{\mkern 1.5mu\overline{\mkern-1.5mu\text#1\mkern-1.5mu}\mkern 1.5mu}

\newtheorem{problem}[theorem]{Problem}

\newtheorem{assumption}{Assumption}

\usepackage[many]{tcolorbox}
\tcbuselibrary{skins}
\newtcolorbox{resp1}[1][]{%
	enhanced jigsaw,%
	colback=gray!5!white,%
	colframe=gray!80!black,%
	size=small,%
	boxrule=1pt,%
	halign title=flush center,%
	coltitle=black,%
	breakable,%
	drop shadow=black!50!white,%
	attach boxed title to top left={xshift=1cm,yshift=-\tcboxedtitleheight/2,yshifttext=-\tcboxedtitleheight/2},%
	minipage boxed title=3cm,%
	boxed title style={%
		colback=white,%
		size=fbox,%
		boxrule=1pt,%
		boxsep=2pt,%
		underlay={%
			\coordinate (dotA) at ($(interior.west) + (-0.5pt,0)$);
			\coordinate (dotB) at ($(interior.east) + (0.5pt,0)$);
			\begin{scope}[gray!80!black]
				\fill (dotA) circle (2pt);
				\fill (dotB) circle (2pt);
			\end{scope}
		}%
	},%
	#1%
}
\title[Robust Control of Uncertain Switched Affine Systems via Scenario Optimization]{Robust Control of Uncertain Switched Affine Systems via Scenario Optimization}
\author{%
 \Name{Negar Monir} \Email{\texttt{\url{s.seyedmonir2@newcastle.ac.uk}}}\\
 \addr School of Computing, Newcastle University, Newcastle upon Tyne, United Kingdom
 \AND
 \Name{Mahdieh S. Sadabadi} \Email{\texttt{\url{mahdieh.sadabadi@manchester.ac.uk}}}\\
 \addr Department of Electrical and Electronic Engineering, The University of Manchester, United Kingdom
 \AND
 \Name{Sadegh Soudjani} \Email{\texttt{\url{sadegh@mpi-sws.org}}}\\
 \addr Max Planck Institute for Software Systems, Kaiserslautern, Germany
}
\begin{document}
\maketitle
\begin{abstract}%
Switched affine systems are often used to model and control complex dynamical systems that operate in multiple modes. However, uncertainties in the system matrices can challenge their stability and performance. This paper introduces a new approach for designing switching control laws for uncertain switched affine systems using data-driven scenario optimization. Instead of relaxing invariant sets, our method creates smaller invariant sets with quadratic Lyapunov functions through scenario-based optimization, effectively reducing chattering effects and regulation error. The framework ensures robustness against parameter uncertainties while improving accuracy. We validate our method with applications in multi-objective interval Markov decision processes and power electronic converters, demonstrating its effectiveness.
\end{abstract}
\section{Introduction}
\noindent\textbf{Motivation:} Switched affine systems play a critical role in various practical applications, including decision-making in Markov Decision Processes (MDPs), power electronics, and control systems in robotics \citep{seatzu2006optimal, albea2015hybrid, beneux2018integral, iervolino2023lyapunov, monir2024lyapunov}. These systems are particularly useful for modeling situations where a system operates under multiple modes governed by distinct switching laws. However, a significant challenge in these applications is managing uncertainty, which can arise from parameter variations, external disturbances, or incomplete model information. This uncertainty greatly affects the stability and robustness of systems, especially in real-world scenarios where precise control and reliable decision-making are essential. Thus, it is crucial to develop robust control strategies that specifically address these uncertainties to ensure optimal and practical implementations in these fields.

\noindent\textbf{Related Works:} Switched affine systems have been gaining attention in the literature, particularly regarding their stability and control \citep{deaecto2010switched, albea2020robust, della2021data, seuret2023practical}. \cite{deaecto2010switched} focus on stability analysis and control design using quadratic Lyapunov functions and optimization techniques to minimize the volume of the invariant set. However, a common limitation in many existing studies is the assumption that system parameters are precisely known without accounting for uncertainties. To address this issue, \cite{albea2020robust} has proposed robust control methods tailored for discrete-time switched affine systems that incorporate parameter uncertainties in system matrices, 
which may lead to chattering effects and increased regulation error. \cite{seuret2023practical} have studied the case where the system model is entirely unknown and employed data-driven methods and experimental data to develop robust switching control laws. Recent studies have also proposed observer-based feedback strategies to improve robustness against uncertainties and disturbances in switched systems \citep{zhang2022output}.

Recently, data-driven methods have emerged as a powerful approach for designing robust control strategies for systems with unknown or uncertain dynamics \citep{de2019formulas,van2020data, berberich2020data}. Recent advances in data-driven synthesis with formal guarantees include building finite abstract models from data \citep{makdesi2021data,kazemi2024data,schon2024bayesian,nazeri2025data}, building continuous models from data \citep{skovbekk2025formal,schon2024data}, or synthesizing control strategies directly without building the model \citep{wu2023risk,kordabad2025data,salamati2024data}.

Scenario-based approaches are data-driven optimizations that have gained traction because they rely minimally on assumptions about data. Unlike other methods that require specific data conditions or model identification, sampled data is used to address robust control and chance-constrained problems, making it ideal for real-world applications. Scenario optimization simplifies robust optimization problems by using sampled scenarios to approximate uncertain constraints \citep{campi2009scenario,campi2021scenario,CAMPI20211,campi2018introduction,garatti2024non}. For instance, scenario programming have been applied to design stabilizing state-feedback laws for switched linear systems using observed trajectories and Lyapunov-based criteria in \cite{wang2024data}. Recent studies have applied these methods in formal verification and synthesis. For instance, \citet{salamati2022data} introduced scenario convex programming for safety verification in stochastic systems, enabling the derivation of robust solutions with probabilistic guarantees. \citet{dietrich2024nonconvex} extended nonconvex scenario programming to address the reachability problem. 
\citet{monir2025logic} and \citet{saoud2024temporal} have applied scenario optimization to study logic-based resilience for dynamical systems against temporal specifications with applications in power systems \citep{monir2025logic} and in water resource recovery facilities \citep{laino2025logic}.
These developments demonstrate the versatility and effectiveness of scenario optimization in the face of uncertainty.

\noindent\textbf{Contribution:} This paper proposes a novel approach to designing switching control laws for switched affine systems with uncertainties in the system matrices. Our approach utilizes a quadratic Lyapunov function and an invariant set of attraction combined with data-driven scenario optimization. This enables us to reduce both chattering effects and regulation errors, which are common drawbacks of existing methods. Additionally, this paper focuses on designing switching laws without auxiliary control inputs, unlike several prior works that employ both control inputs and switching laws to stabilize the system. We also demonstrate the applicability of our proposed method to various practical systems, including Multi-Objective Interval MDPs (MOIMDPs) and power electronic converter systems, showcasing its effectiveness in addressing real-world challenges.

\noindent\textbf{Outline:}
This paper is organized as follows. Section \ref{sec pre} outlines the background and problem formulation. Section \ref{sec switch law design} presents a method for designing a robust switching control law using Lyapunov functions and scenario optimization. Section \ref{sec app} validates this method with case studies on MOIMDPs and power electronic converters. Finally, Section \ref{sec conc} summarizes the contributions and future research directions.

\section{Preliminaries and Problem Setup} \label{sec pre}
\noindent\textbf{Notations:}
$\N$,  $\R$, $\R^{n}$, and $\R^{n \times m}$ denote, respectively, the sets of natural numbers {including zero}, real numbers, the n-dimensional Euclidean space, and the set of $n \times m$ real matrices. $\N_{\mathcal{P}}$ denotes the set $\{1,2,\dots, \mathcal{P}\}$ and $0_{n\times m}$ indicates the $n\times m$ matrix of zeros. For any $M \in \R^{n \times n}$, the notation $M>0$ ($M<0$) means $M$ is symmetric and positive (negative) definite. $\det(M)$ represents the determinant of $M$. For real vectors or matrices, $\left({ }^{\top}\right)$ refers to their transpose. For symmetric matrices, $(*)$ denotes each of their symmetric blocks. The interior of a set is represented by $\operatorname{int}(\cdot)$. The concatenation of two matrices $A$ and $B$ of appropriate dimensions in a row is denoted by $[A,B]$ and in a column by $[A;B]$. The convex combination of matrices with the same dimension is denoted as $X_\lambda=\sum_{i \in \mathbb{N}_{\mathcal{P}}} \lambda_i X_i$ with $\lambda=\left[\lambda_{1}, \ldots, \lambda_{\mathcal{P}}\right] \in \Lambda_{\mathcal{P}}$ where $\Lambda_{\mathcal{P}}:=\left\{\left[\lambda_{1}, \ldots, \lambda_{\mathcal{P}}\right] \mid \lambda_{i}\ge 0,\sum_{i \in \mathbb{N}_{\mathcal{P}}} \lambda_{i}=1\right\}$ is the unitary simplex. A square matrix is said to be Schur stable if all its eigenvalues are inside the unit circle on the complex plane.

\noindent\textbf{System Description:} Consider the discrete-time switched affine system defined by 
\begin{equation}
    x_{k+1} = A_{\pi} x_k + B_{\pi},\quad{k \in \mathbb{N},} \label{eq. dt SAS}
\end{equation}
where $x_k: \N \ra \R^n$ is the state at time $k$ and $\pi \in \N_\mathcal{P}$ is the switching law to be designed and can change depending on the state. The matrices $A_\pi$ and $B_\pi$ are uncertain but belong respectively to the intervals $ [\underbar{A}_\pi,\overbar{A}_\pi]$ and $ [\underbar{B}_\pi,\overbar{B}_\pi]$. Matrices $A_\pi$ and $B_\pi$ can be rewritten as $A_\pi = \hat{A}_\pi + \Delta A_\pi$ and $B_\pi = \hat{B}_\pi + \Delta B_\pi$, in which $\hat{A}_\pi$ and $\hat{B}_\pi$ are nominal matrices. $\hat{A}_\pi$ and $\hat{B}_\pi$ can be selected from the intervals $[\underbar{A}_\pi,\overbar{A}_\pi]$ and  $[\underbar{B}_\pi,\overbar{B}_\pi]$. Here, we select the midpoint of the interval, i.e. $\hat{A}_\pi = \frac{1}{2} (\underbar{A}_\pi + \overbar{A}_\pi)$, $\hat{B}_\pi = \frac{1}{2} (\underbar{B}_\pi + \overbar{B}_\pi)$.
\begin{remark}
    The above model class includes switched affine systems with parametric uncertainty. Any such systems can be over-approximated with the system in \eqref{eq. dt SAS} by taking the element-wise optimal values with respect to the uncertain parameters.
\end{remark}
Finally, the following assumption will be considered to define the set of allowable operating points for system \eqref{eq. dt SAS}, which is inspired by \citep{deaecto2016stability, albea2020robust}.
\begin{assumption}
    The desired operating point \( x_e \) belongs to $\f$, which is defined as
    \begin{equation}
        \f := \{x_e \in \R^n, \: \exists \lambda \in \Lambda_{\mathcal{P}}, \:(\hat{A}_\lambda - I) x_e + \hat{B}_\lambda = 0\}, \label{eq. operating points}
    \end{equation}
     where ${\Ad}_{\lambda}:=\sum_{\pi \in \N_{\mathcal{P}}}\lambda_{\pi}{\Ad}_{\pi}$, $\hat{B}_{\lambda}=\sum_{\pi \in \N_{\mathcal{P}}}\lambda_{\pi}\hat{B}_{\pi}$ and ${\Ad}_{\lambda}$ is assumed to be Schur stable.
    \label{ass. operating set}
\end{assumption}
\noindent\textbf{Error Dynamics:} Introducing the error variable $\xi_k := x_k - x_e$, the resulting error dynamics are
\begin{equation}
    \xi_{k+1} = A_{\pi} \xi_k + L_{\pi}, \label{eq. error dy}
\end{equation}
where $L_{\pi}= (A_{\pi} - I)x_e + B_{\pi}$. Hence, $\underbar{L}_{\pi}= (\underbar{A}_{\pi} - I)x_e + \underbar{B}_{\pi}$ and $\overbar{L}_{\pi}= (\overbar{A}_{\pi} - I)x_e + \overbar{B}_{\pi}$. This can be rewritten as $L_{\pi} = \hat{L}_{\pi} + \Delta L_{\pi}$, where $\hat{L}_{\pi}= (\hat{A}_{\pi} - I)x_e + \hat{B}_{\pi}$ and $\Delta L_{\pi} = \Delta A_{\pi} x_e + \Delta B_{\pi}$. 

\noindent\textbf{Invariant Set of Attraction:} The goal is to design a state-dependant switching function $\pi(\xi)$
so that ideally, \(\xi_k \rightarrow 0\) as \(k \rightarrow \infty\) for any initial condition \(x_0 \in \mathbb{R}^n\). However, in the context of uncertain discrete-time systems, achieving this goal is generally impossible. Therefore, we aim to design \(\pi(\xi)\) that switches between the elements of \(\N_{\mathcal{P}}\) to steer the trajectories of the system toward an \emph{invariant set of attraction} defined next, which encompasses the desired operating point. 

\begin{definition} \textbf{(Invariant Set of Attraction)}
A set $\Omega \subset \R^{n}$ is an invariant set of attraction of the system \eqref{eq. error dy} governed by the switching policy $\pi(\xi)$, if there exists a Lyapunov function $V:\R^{n}\rightarrow \R$ such that the following conditions are simultaneously fulfilled: (i)~$0_{n \times 1} \in \Omega$, (ii)~ if ${\xi}_k \notin \Omega$, then $\Delta V({\xi}):= V({\xi}_{k+1})-V({\xi}_k) < 0$, and (iii)~if ${\xi}_k \in \Omega$, then ${\xi}_{k+1} \in \Omega$. These conditions are denoted as specification $\Psi$ for $\Omega$. We denote by $\Omega \models \Psi$ whenever $\Omega$ satisfies these three conditions.
\label{def. roa}
\end{definition}

The existence of such a Lyapunov function ensures that, under the switching policy $\pi$, the error will converge to the set $\Omega$ from any initial error $\xi_0$ and remain within $\Omega$ thereafter. This property classifies the system as \emph{practically stable}. We employ a quadratic candidate Lyapunov function
\begin{equation}
    V(\xi) := d + {2h^\top}\xi +\xi^\top F\xi, \label{eq. lyapunov func}
\end{equation}
where $h \in \R^n$, $0<F \in \R^{n \times n}$, and $h^\top F^{-1} h = d \in \R$. 
Obviously, $V$ is a convex function such that \( V(\xi) > 0 \) for all \( \xi \) except at \( \hat{\xi} = -F^{-1} h \), where \( V(\xi) = 0 \). This naturally leads to the following minimum-type state feedback switching control~law
\begin{equation}
    \pi(\xi) = \argmin_{\pi \in \N_{\mathcal{P}}} V(A_\pi \xi + L_\pi). \label{eq. switch law}
\end{equation}
Furthermore, we consider the level set of $V$ with $\rho>0$ defined as 
\begin{equation}
    \Omega(\rho) = \{\xi \in \R^n: V(\xi) \le \rho \}. \label{eq. roa}
\end{equation}
We aim to identify a minimal-volume invariant set of attraction to ensure that error trajectories starting outside will eventually enter, including the origin, and prevent solutions from exiting once inside. In the steady state, error trajectories remain within this set, and smaller sets help reduce chattering amplitude. We thus define the following problem:
\begin{resp1}
\begin{problem}
Given the error dynamics in \eqref{eq. error dy} and the desired operating point $x_e$, find a set $\Omega(\rho)$ as in Definition \ref{def. roa} that satisfies the specification $\Psi$ with a confidence of at least $(1-\beta) \in (0,1)$ and the associated control law $\pi(\xi)$.
\label{prblm. roa}
\end{problem}
\end{resp1}
\section{Data-driven Design of Robust Switching Control Law} \label{sec switch law design}
In this section, we will present an algorithm designed to address Problem \ref{prblm. roa}.

\subsection{Stability Analysis} \label{subsec. roa design}
The following theorem outlines the ellipsoidal set of attraction with the minimum volume, based on the Lyapunov function in \eqref{eq. lyapunov func}, the switching control law in \eqref{eq. switch law}, and Definition \ref{def. roa}.
\begin{theorem}
    Consider the error dynamics of the switched affine system in \eqref{eq. error dy}. From the optimal solution $F > 0$, $W > 0$, and $h$ of the robust convex programming 
     \begin{align}  \boldsymbol{RCP}^{\rom{1}}\!\!:\!\!\! \inf_{F>0, W>0, h}\!\! \Big\{\! -\ln(\det(W)): & Q_\lambda > W,\, \begin{bmatrix}
        Q_\lambda & \varrho_\lambda \\
        * & 1-c_\lambda
    \end{bmatrix}\!>\!
    0, \begin{bmatrix}
        Q_\lambda & \varrho_\lambda \\
        * & 1
    \end{bmatrix} \!>
    \!0,\nonumber\\
    &\forall A_\pi \in [\underbar{A}_\pi, \overbar{A}_\pi],\forall L_\pi \in [\underbar{L}_\pi, \overbar{L}_\pi]\!\Big\}, \label{eq. rbst cvx prgrm}
\end{align}
with $Q_\lambda = \Sigma_\pi \lambda_\pi(F -  A_\pi^\top F A_\pi)$, $\varrho_\lambda = \Sigma_\pi \lambda_\pi(A_\pi^\top h + A_\pi^\top F L_\pi -h)$ and $c_\lambda = \Sigma_\pi \lambda_\pi(2h^\top L_\pi + L_\pi^\top F L_\pi)$, determine the quadratic Lyapunov function $V(\xi)$ as defined in \eqref{eq. lyapunov func}. The state-dependent switching control law \eqref{eq. switch law} ensures the existence of the set $\mathcal{X}$ centered at $\xi_o = Q_\lambda^{-1} \varrho_\lambda$ for all $A_\pi \in [\underbar{A}_\pi, \overbar{A}_\pi]$ and $L_\pi \in [\underbar{L}_\pi, \overbar{L}_\pi]$ as
\begin{align}
    \mathcal{X} =\{&\xi \in \R^{n}: {(\xi-\xi_o)^\top} {W} (\xi-\xi_o) \leq 1\},   \label{eq. X set}
\end{align}
which satisfies the first two conditions of Definition \ref{def. roa} with minimum volume. Moreover, the center $\hat{\xi}$ of the Lyapunov function \eqref{eq. lyapunov func} lies in the interior of \eqref{eq. X set}.
    \label{thm. set of attraction}
\end{theorem}
\begin{proof}
    The candidate Lyapunov function \eqref{eq. lyapunov func} is strictly convex, except at $\hat{\xi} = -F^{-1} h$, where it achieves its global minimum value of $V(\hat{\xi}) = 0$. By applying the control switching law to the trajectories of error dynamics \eqref{eq. error dy}, it is evident that
    \begin{equation}
    \Delta V(\xi) = V(\xi_{k+1}) - V(\xi_k) = \min_\pi V(A_\pi \xi + L_\pi)-V(\xi) \le \underbrace{\Sigma_\pi \lambda_\pi V(A_\pi \xi + L_\pi)-V(\xi)}_{f(\xi)}. \label{eq. Delta V}
\end{equation}
The right-hand side of \eqref{eq. Delta V} allows us to define the following quadratic function
\begin{align}
    f(\xi) &= \Sigma_\pi \lambda_\pi V(A_\pi \xi + L_\pi)-V(\xi) \nonumber \\
    &= \Sigma_\pi \lambda_\pi \{ \xi^\top(A_\pi^\top F A_\pi - F) \xi + 2 (h^\top A_\pi + L_\sigma^\top F A_\sigma - h^\top)\xi + (2h^\top L_\pi + L_\pi^\top F L_\pi)\} \nonumber\\
    &= -\xi^\top \underbrace{\{\Sigma_\pi \lambda_\pi(F -  A_\pi^\top F A_\pi)\}}_{Q_\lambda} \xi + 2 \{\underbrace{\Sigma_\pi \lambda_\pi(A_\pi^\top h + A_\pi^\top F L_\pi -h)}_{\varrho_\lambda(h)}\}^\top \xi  \nonumber \\
    & \quad + \underbrace{\Sigma_\pi \lambda_\pi(2h^\top L_\pi + L_\pi^\top F L_\pi)}_{c_\lambda(h)} \nonumber \\
    &= -(\xi - Q_\lambda^{-1} \varrho_\lambda)^\top Q_\lambda (\xi - Q_\lambda^{-1} \varrho_\lambda) + c_\lambda + \varrho_\lambda^\top Q_\lambda^{-1} \varrho_\lambda. \label{eq. f(xi)}
\end{align}
If the robust convex programming problem in \eqref{eq. rbst cvx prgrm} admits a solution, then there exist symmetric positive definite matrices $F$ and $W$ such that $Q_\lambda = \Sigma_\pi \lambda_\pi(F -  A_\pi^\top F A_\pi) > W > 0, \, \forall A_\pi \in [\underbar{A}_\pi,\overbar{A}_\pi]$. 
Given that \( Q_{\lambda} > W > 0 \) and by applying the second constraint from \eqref{eq. rbst cvx prgrm}, the set defined in \eqref{eq. X set} ensures that \( \Delta V(\xi) < 0 \) for all \( \xi \notin \mathcal{X} \). We can demonstrate this based on equations \eqref{eq. Delta V}, \eqref{eq. f(xi)} with  \( -Q_{\lambda} < -W < 0 \) and using Schur complements \citep{boyd2004convex} as follows
\begin{align}
    \Delta V(\xi) &\le -(\xi - Q_\lambda^{-1} \varrho_\lambda)^\top Q_\lambda (\xi - Q_\lambda^{-1} \varrho_\lambda) + c_\lambda + \varrho_\lambda^\top Q_\lambda^{-1} \varrho_\lambda \nonumber\\
    &<  -(\xi - Q_\lambda^{-1} \varrho_\lambda)^\top W (\xi - Q_\lambda^{-1} \varrho_\lambda) + c_\lambda + \varrho_\lambda^\top Q_\lambda^{-1} \varrho_\lambda \nonumber \\
    &< 0, \, \forall \xi \notin \mathcal{X}, \nonumber \\
    \Leftarrow \quad &c_\lambda + \varrho_\lambda^\top Q_\lambda^{-1} \varrho_\lambda < 1 \iff \begin{bmatrix}
        Q_\lambda & \varrho_\lambda \\
        \varrho_\lambda^\top & 1-c_\lambda
    \end{bmatrix} > 0, \forall A_\pi \in [\underbar{A}_\pi, \overbar{A}_\pi], \forall L_\pi \in [\underbar{L}_\pi, \overbar{L}_\pi] .\label{eq. c lambda}
\end{align}
In addition, the third constraints in \eqref{eq. rbst cvx prgrm} ensures that origin is inside the set \eqref{eq. X set}. We can show this using Schur complement \citep{boyd2004convex} as follows
\begin{align}
    0 \in \mathcal{X} \iff  &(0 - Q_\lambda^{-1} \varrho_\lambda)^\top W (0 - Q_\lambda^{-1} \varrho_\lambda) < (0 - Q_\lambda^{-1} \varrho_\lambda)^\top Q_\lambda (0 - Q_\lambda^{-1} \varrho_\lambda) < 1 \nonumber \\
    \Leftarrow \,\,\, &\varrho_\lambda^\top Q_\lambda^{-1} \varrho_\lambda < 1 \iff \begin{bmatrix}
        Q_\lambda & \varrho_\lambda \\
        \varrho_\lambda^\top & 1
    \end{bmatrix} > 0, \forall A_\pi \in [\underbar{A}_\pi, \overbar{A}_\pi], \forall L_\pi \in [\underbar{L}_\pi, \overbar{L}_\pi]. \label{eq. ro lambda quadratic}
\end{align}
The constraints \eqref{eq. c lambda} and \eqref{eq. ro lambda quadratic} implies that $c_\lambda \in (0,1)$ as we know that $0<\varrho_\lambda^\top Q_\lambda^{-1} \varrho_\lambda < 1$. 
The solution to the problem in \eqref{eq. rbst cvx prgrm} ensures the minimization of the volume of the set \(\mathcal{X}\) defined in \eqref{eq. X set}. This volume is inversely proportional to \(\sqrt{\operatorname{det}(W)}\) \citep{boyd2004convex}. Furthermore, we have \(\hat{\xi} = -{F}^{-1}{h} \in \operatorname{int}(\mathcal{X})\). From \eqref{eq. lyapunov func} and \eqref{eq. f(xi)}, it follows ${f(\hat{\xi})=\sum_{\pi \in \mathbb{M}^{\prime}} \lambda_{\pi} V\left(A_{\pi} \hat{\xi}+L_{\pi}\right)>0}$, being $V(\hat{\xi})=0$ and $V(\xi)>0, \forall \xi \neq \hat{\xi}$. From \eqref{eq. Delta V} and \eqref{eq. c lambda}, we have
\begin{align*}
    &f(\hat{\xi}) > 0 \, \Rightarrow -(\hat{\xi} - Q_\lambda^{-1} \varrho_\lambda)^\top Q_\lambda (\hat{\xi} - Q_\lambda^{-1} \varrho_\lambda) < c_\lambda + \varrho_\lambda^\top Q_\lambda^{-1} \varrho_\lambda < 1.
\end{align*}
Finally, from the first constraint in \eqref{eq. rbst cvx prgrm}, it implies that
\begin{equation*}
    -(\hat{\xi} - Q_\lambda^{-1} \varrho_\lambda)^\top Q_\lambda (\hat{\xi} - Q_\lambda^{-1} \varrho_\lambda) < -(\hat{\xi} - Q_\lambda^{-1} \varrho_\lambda)^\top W (\hat{\xi} - Q_\lambda^{-1} \varrho_\lambda) < 1,
\end{equation*}
which concludes the proof.
\end{proof}
\begin{remark}
    In Theorem \ref{thm. set of attraction}, the set $\mathcal{X}$ in \eqref{eq. X set} is an ellipsoid centered at $\xi_o = Q_\lambda^{-1} \varrho_\lambda$, with $Q_\lambda = \Sigma_\pi \lambda_\pi(F -  A_\pi^\top F A_\pi)$ and $\varrho_\lambda = \Sigma_\pi \lambda_\pi(A_\pi^\top h + A_\pi^\top F L_\pi -h)$. Since the constraints in the optimization problem \eqref{eq. rbst cvx prgrm} will be satisfied for all $A_\pi \in [\underbar{A}_\pi, \overbar{A}_\pi]$ and $L_\pi  \in [\underbar{L}_\pi, \overbar{L}_\pi]$, for any $A_\pi$ and $L_\pi$ within the specified intervals, there exists a center $\xi_o$ that defines the set $\mathcal{X}$ to satisfy the given constraints.
\end{remark}
Theorem \ref{thm. set of attraction} helps us identify a set of attraction that meets the first two conditions in Definition~\ref{def. roa} by addressing the optimization problem in \eqref{eq. rbst cvx prgrm}. To establish an invariant set of attraction, we must identify additional conditions to satisfy the final criterion in Definition \ref{def. roa}. Thus, the following theorem is presented. 
\begin{theorem}
Let matrices $F > 0$, $W > 0$, and the vector
$h$ denote the solutions to the optimization problem defined in \eqref{eq. rbst cvx prgrm}.
The solution ${\rho}$ of the following robust convex program
\begin{align} \label{eq. opt ro}
    \boldsymbol{RCP}^{\rom{2}}\!\!:\!\!\!\!\! \inf_{\rho>0,\tau>0} & \!\!\left\{\!\! \rho\!:\!\!\!   \begin{bmatrix}
        \rho+\tau \xi_o^\top W \xi_o-\tau \!&\! h^\top \!&\! -\tau \xi_o^\top W \\
        * \!&\! {F} \!&\! {F} \\
        * \!&\! * \!&\! \tau W
    \end{bmatrix}\!\!  >\! 0, \forall A_\pi, \, \forall L_\pi \in [\underbar{A}_\pi, \overbar{A}_\pi],\, [\underbar{L}_\pi, \overbar{L}_\pi]\!\!\right\}\!,
\end{align}
with $\xi_o = Q_\lambda^{-1} \varrho_\lambda$, where $Q_\lambda = \Sigma_\pi \lambda_\pi(F -  A_\pi^\top F A_\pi)$ and $\varrho_\lambda = \Sigma_\pi \lambda_\pi(A_\pi^\top h + A_\pi^\top F L_\pi -h)$, provides the ellipsoidal invariant set of attraction $\Omega(\rho) \subseteq \mathcal{X}$ of minimum volume, expressed according to \eqref{eq. roa} and satisfying all the three conditions of Definition \ref{def. roa}. 
\label{thm. ro opt}
\end{theorem}
\begin{proof}
We draw inspiration from the steps taken in \citep{deaecto2016stability}.
    Let us consider the Lyapunov function \eqref{eq. lyapunov func} with $F$, $h$, and $d$, calculated by applying Theorem~\ref{thm. set of attraction}. To determine the invariant set of attraction with the minimum volume containing $\mathcal{X}$, we consider the following condition for all $A_\pi \in [\underbar{A}_\pi, \overbar{A}_\pi]$ and $L_\pi \in [\underbar{L}_\pi, \overbar{L}_\pi]$
\begin{align}
\rho&=  \max _{\xi \in \mathcal{X}} V(\xi)=\inf _{\rho>0}\{\rho: V(\xi)<\rho, \forall \xi \in \operatorname{int}(\mathcal{X})\} \nonumber \\
&= \inf _{\rho>0}\left\{\rho: d-\rho+2 h^\top \xi + \xi^\top F \xi <0, \forall \xi \in \R^{n}, \,\st \, -1 + (\xi-\xi_o)^\top W (\xi-\xi_o) <0\right\}\!, \label{eq. ro 1}
\end{align}
where $\mathcal{X}$ is the set given in \eqref{eq. X set}. By applying the $\mathcal{S}$-procedure \citep{boyd2004convex}, \eqref{eq. ro 1} becomes as follows
\begin{align}
\rho\!= \!\!\!\!\! \inf _{\substack{\rho>0, \tau>0\\
}} \!\{\! &\rho\!:\! h^\top (F)^{-1} h+\tau - \tau \xi_o^\top W \xi_o +2 {(h\!+\!\tau W \xi_o)^\top} \xi + {\xi^\top}(F\!-\!\tau W) \xi \!<\! \rho, \forall \xi \in \mathbb{R}^{n}\} . \label{eq. ro 2}
\end{align}
The proof boils down to showing that the problem in \eqref{eq. ro 2} is feasible. From the constraint of the optimization problem \eqref{eq. opt ro}, and through the Schur complements, it follows
\begin{align*}
    \begin{bmatrix}
        F & F \\
        F & \tau W
    \end{bmatrix} > 0 \iff F-\tau W<0,
\end{align*}
thus the quadratic function in \eqref{eq. ro 2}, say $g(\xi)=h^\top F^{-1} h+\tau - \tau \xi_o^\top W \xi_o - \rho +2 \transpose{(h+\tau W \xi_o)} 
\xi + \xi^\top(F-\tau W) \xi$, is strictly concave. Its extreme point occurs for $\xi^*=-(F-\tau W)^{-1} (h + \tau W \xi_o)$, where the function attains its maximum value, $g(\xi^*)= - \transpose{(h+\tau W \xi_o)} (F-\tau W) ^ {-1} (h+\tau W \xi_o) + h^\top F^{-1} h+\tau - \tau {\xi_o^\top} W \xi_o -\rho$. From the constraint of the optimization problem \eqref{eq. opt ro}, it is also
\begin{align*}
    \begin{bmatrix}
        (\rho+\tau {\xi_o^\top}{W}{\xi}_o-\tau-\transpose{(h)} (F)^{-1} h) & -\transpose{(h+\tau W {\xi}_o)} \\
-{(h+\tau W {\xi}_o)} & -(F-\tau W)
    \end{bmatrix} > 0,
\end{align*}
which, through simple Schur complement calculations, implies that
\begin{align*}
    g({{\xi}}^*)= & - \transpose{(h+\tau W {\xi}_o)} (F-\tau W) ^ {-1} (h+\tau W {\xi}_o) + \transpose{h} F^{-1} h+\tau - \tau \transpose{{\xi}_o} W {\xi}_o  - \rho < 0.
\end{align*}
As a result, it is $g({{\xi}})<0, \forall {\xi} \in \mathbb{R}^{n}, \forall A_\pi \in [\underbar{A}_\pi, \overbar{A}_\pi]$, and $\forall L_\pi \in [\underbar{L}_\pi, \overbar{L}_\pi]$, and the problem in \eqref{eq. ro 2} admits a solution.
Finally, it is easy to verify that $\Omega(\rho)$ is an invariant set of attraction, i.e., it satisfies all the conditions in Definition \ref{def. roa}, being a $\mathcal{X} \subseteq \Omega(\rho)$ and $V(\xi)$, determined by applying Theorem \ref{thm. ro opt}, is a valid Lyapunov function for the system \eqref{eq. error dy} (i.e., its level sets are invariant).
\end{proof}

Here, the main challenge is to solve two robust convex programming problems outlined in \eqref{eq. rbst cvx prgrm} and \eqref{eq. opt ro}, which will be addressed in Section \ref{subsec scenario}.
\begin{remark}
The optimization problems in equations \eqref{eq. rbst cvx prgrm} and \eqref{eq. opt ro} define a robust convex program with an infinite number of Linear Matrix Inequality (LMI) constraints, as the inequalities must hold for all admissible parameters \( A_\pi \) and \( L_\pi \) within their specified bounds. To simplify this problem, we will use the scenario approach in the next section, which approximates the original problem by applying constraints to a finite set of sampled scenarios. 
\end{remark}
\subsection{Scenario Optimization to Design the Invariant Set of Attraction  \label{subsec scenario}}
This section addresses proposed robust convex programming problems through a data-driven wait-and-judge scenario approach \citep{campi2018wait, CAMPI20211}, where the robust optimization problems in \eqref{eq. rbst cvx prgrm} and \eqref{eq. opt ro} can be solved with a certain confidence by taking random samples from the uncertain variables.

The optimization problem \eqref{eq. rbst cvx prgrm} can be reformulated as a convex scenario program. We consider a uniform distribution over the space $\mathcal{W} = [\underbar{A}_\pi, \overbar{A}_\pi] \times [\underbar{L}_\pi, \overbar{L}_\pi]$ and obtain $N$ sampled scenarios 
$\omega^i = (A_\pi,L_\pi)^i = (A_\pi^i,L_\pi^i)$, $i\in\{1,2,\ldots, N\}$. Then, the optimization \eqref{eq. rbst cvx prgrm} can be approximated as 
\begin{align}
   \boldsymbol{SP_N^{\rom{1}}\!:\!} \inf_{z} \Biggl\{\!-\!\ln(\det(W))\!:\! &\underbrace{Q_\lambda^i \!>\! W, \begin{bmatrix}
        Q_\lambda^i & \varrho_\lambda^i \\
        (\varrho_\lambda^i)^\top & 1-c_\lambda^i
    \end{bmatrix} \!>\!
    0,\!  \begin{bmatrix}
        Q_\lambda^i & \varrho_\lambda^i \\
        (\varrho_\lambda^i)^\top & 1
    \end{bmatrix} \!>\!
    0, F\!>\!0, W\!>\!0,}_{\text{Set of Constraints } \mathcal{Z}}
    \forall i\!\in\!\N_N\! \Biggl\}, \label{eq. rbst cvx prgrm scenario}
\end{align}
where $Q_\lambda^i = \Sigma_\pi \lambda_\pi(F -  (A_\pi^i)^\top F A_\pi^i)$, $\varrho_\lambda^i = \Sigma_\pi \lambda_\pi((A_\pi^i)^\top h + (A_\pi^i)^\top F L_\pi^i -h)$, $c_\lambda^i = \Sigma_\pi \lambda_\pi(2h^\top L_\pi^i + (L_\pi^i)^\top F L_\pi^i)$ and $z = \{F, W, h\}$ is the set of variables. Denote by $\mathcal V_1(z)$ the constraint violation probability with respect to the probability measure put on the uncertainty space $\mathcal W$ when the optimization variables are set to $z$. Let \( z_N^* \) represent the optimal solution to equation \eqref{eq. rbst cvx prgrm scenario}. 
The solution is considered robust at the robustness level \( \varepsilon_1 \) if \( \mathcal{V}_1(z_N^*) \leq \varepsilon_1 \).

Similarly, we approximate the optimization \eqref{eq. opt ro} with a convex scenario program. Consider the uniform distribution over \(\mathcal{W}\) and obtain \(M\) new sampled scenarios \(w^j = (A_\pi,L_\pi)^j = (A_\pi^j,L_\pi^j)\), $j\in\{1,2,\ldots,M\}$. The optimization \eqref{eq. rbst cvx prgrm scenario} can then be approximated as
\begin{align}
\label{eq. opt ro scenario}
     \boldsymbol{SP_M^{\rom{2}}\!: }\! \inf_{{y}} & \left\{\! \rho\!:\!\!   \underbrace{\begin{bmatrix}
        \rho+\tau (\xi_o^j)^\top W \xi_o^j-\tau & h^\top & -\tau (\xi_o^j)^\top W \\
        {h} & {F} & {F} \\
        -\tau W (\xi_o^j)^\top & F & \tau W
    \end{bmatrix}\!\!>\! 0, \rho>\!0, \tau>\!0,}_{\text{Set of Constraints } \mathcal{Y}} \forall j \in \mathbb{N}_M\right\}, 
\end{align}
with $\xi_o^j = (Q_\lambda^j)^{-1} \varrho_\lambda^j$, where $Q_\lambda^j = \Sigma_\pi \lambda_\pi(F -  (A_\pi^j)^\top F A_\pi^j)$, $\varrho_\lambda^j = \Sigma_\pi \lambda_\pi((A_\pi^j)^\top h + (A_\pi^j)^\top F L_\pi^j -h)$ and $y = \{\rho, \tau\}$ is the set of variables. Denote by $\mathcal V_2(y)$ the constraint violation probability with respect to the probability measure put on the uncertainty space $\mathcal W$ when the optimization variables are set to $y$. Let \( y_M^* \) represent the optimal solution to equation \eqref{eq. opt ro scenario}. If \( \mathcal{V}_2(y_M^*) \leq \varepsilon_2 \), the solution is considered robust at the robustness level \( \varepsilon_2 \).

A scenario \(\omega\) is considered a \textbf{\emph{support scenario}} if its removal changes the solution to optimization problem. Denote by \(s_{1N}^*\) and \(s_{2M}^*\) the number of support scenarios of \eqref{eq. rbst cvx prgrm scenario} and \eqref{eq. opt ro scenario}, respectively.
Following the \emph{wait-and-judge} paradigm described by \cite{campi2018wait}, we have the following theorem, which describes how to compute robustness levels \( \varepsilon_1, \varepsilon_2 \) with a certain confidence by identifying the number of support scenarios.
\begin{theorem}
\label{thm. conf total}
Assume the optimizations in \eqref{eq. rbst cvx prgrm} and \eqref{eq. opt ro} are feasible.   For a given confidence value $\beta\in(0,1)$, we have that 
    $$\mathbf{P}^{N+M}\left\{\mathcal{V}_1\left(z_{N}^{*}\right)>\varepsilon_1\left(s_{1N}^{*}\right) \text{ and }\mathcal{V}_2\left(y_{M}^{*}\right)>\varepsilon_2\left(s_{2M}^{*}\right)\right\} \leq \beta,$$
    where $\varepsilon_i:=1-\iota_i$, $i=1,2$, and $\iota_1,\iota_2$ are the unique solutions of the following equation for $\iota$
    \begin{equation*}
    \frac{\eta}{N_c+1} \sum_{q=k}^{N_c}\binom{q}{k} \iota^{q-k}-\binom{N_c}{k} \iota^{N_c-k}=0,
\end{equation*}
by replacing $(\eta,N_c,k)$ respectively with $(\beta_1,N,s_{1N}^*)$ and $(\beta_2,M,s_{2M}^*)$, and where $\beta = \beta_1+\beta_2$.
\end{theorem}
\begin{proof}
Following the approach of \cite{campi2018wait}, we get the two inequalities $\mathbf{P}^{N}\left\{\mathcal{V}_1\left(z_{N}^{*}\right)>\varepsilon_1\left(s_{1N}^{*}\right)\right\} \leq \beta_1$  and
    $\mathbf{P}^{M}\left\{\mathcal{V}_2\left(y_{M}^{*}\right)>\varepsilon_2\left(s_{2M}^{*}\right)\right\} \leq \beta_2$. These two inequalities can then be combined on their product space by applying Bool's inequality to the complement of the involved events \citep{galambos1977bonferroni}.
\end{proof}
 
Note that if the robust convex programming problems in \eqref{eq. rbst cvx prgrm} and \eqref{eq. opt ro} are feasible, then the corresponding scenario programs in \eqref{eq. rbst cvx prgrm scenario} and \eqref{eq. opt ro scenario} will also be feasible. The converse is not necessarily true due to the reduction in the number of constraints. The above theorem enables us to have a trade-off between the two confidence values $\beta_1,\beta_2$ attributed to the two constraint violation probabilities. The results of Theorems \ref{thm. set of attraction}, \ref{thm. ro opt} and \ref{thm. conf total}, along with their scenario programming reformulation in this section, are combined into Algorithm \ref{alg} to design a robust switching control law $\pi$.

\begin{algorithm}[t!]
    \caption{Robust Switching Control Law Design}\label{alg}
    \vspace{-4mm}
    \begin{multicols}{2}
    \begin{algorithmic}[1]
        \small\REQUIRE Interval uncertainty $[\underbar{A}_\pi,\overbar{A}_\pi]$ and $[\underbar{B}_\pi,\overbar{B}_\pi]$ for system matrices in \eqref{eq. dt SAS}, confidence level $\beta$
        \STATE Define $\f$ as in \eqref{eq. operating points} based on the chosen $\hat{A}_\pi, \hat{B}_\pi \in [\underbar{A}_\pi,\overbar{A}_\pi], [\underbar{B}_\pi,\overbar{B}_\pi]$  
        \STATE Compute $x_{e}' := \argmin_{x'} \{\|x' - x\|, x\in\f\}$
        \STATE Define the confidence level for each optimization of $\boldsymbol{SP_N^{ \rom{1}}}$ and $\boldsymbol{SP_M^{ \rom{2}}}$ in a way that $\beta_{{1}} + \beta_{{2}} = \beta$.
        \STATE Decide on number of samples $N$ based on $\beta_{1}$ and Theorem \ref{thm. conf total}
        \STATE Solve $\boldsymbol{SP_N^{\rom{1}}}$ in \eqref{eq. rbst cvx prgrm scenario} to get $F$, $W$, and $h$
        \STATE Set $F$, $h$, and $d = h^\top F^{-1} h$ in Lyapunov function \eqref{eq. lyapunov func} and get $V(\xi)$
        \STATE Select the number of samples $M$ using $\beta_{{2}}$ and Theorem \ref{thm. conf total}
        \STATE Solve $\boldsymbol{SP_M^{\rom{2}}}$ in \eqref{eq. opt ro scenario} to get $\rho$
        \STATE Set $\rho$ in \eqref{eq. roa} and get invariant set of attraction $\Omega(\rho)$
        \STATE Compute $\mathcal{G}$ as the smallest set such that $\Omega + x_{e}'\subset \mathcal{G}$ and $x_{e}\in\mathcal{G}$
        \STATE Set $x_0 = 0 $ and $\xi_0=x_0-x_e$\\
        \STATE\For{$k\in\N$}
        {Compute ${\pi}_k(E_k)$ using \eqref{eq. switch law}\;
        Update $\xi_{k+1}$ using \eqref{eq. error dy} and by applying the upper-bound and lower-bound of uncertainty}
        \ENSURE Set $\mathcal{G}$ and switching control law ${\pi}$ 
    \end{algorithmic}
    \end{multicols}
    \vspace{-3mm}
\end{algorithm}

\section{Case studies} \label{sec app}
Switched affine systems with interval uncertainty have various applications. This study focuses on two applications: MOIMDPs and power electronic systems. We also compare the quality of the solution on a numerical example with an alternative approach from the literature. 
In all the examples provided, the optimization conditions are verified using the \textsf{CVX} package in \textsf{MATLAB} \citep{cvx}. The computations are performed on a \textsf{MacBook} with the \textsf{M2~chip} and \textsf{16GB} of memory.

\subsection{MOIMDPs}
Designing policies to optimize multiple objectives in IMDPs can be modeled with uncertain switched affine systems \citep{monir2024lyapunov}, which is performed through value iteration algorithms. The values $W$ that give the objectives as a function of state satisfy
$W_{k+1}= A_{\pi} W_{k} + B_{\pi}$, where $\pi$ is the policy at time $k$, matrices $A_{\pi}$ depend on the transition probabilities, and matrices $B_{\pi}$ contain the objectives as a function of states. 
The target values $W_{tar} := \begin{bmatrix} w_{tar}^1; w_{tar}^2; \ldots; w_{tar}^q \end{bmatrix}$ for all objectives can be selected from the desired operating points presented in \eqref{eq. operating points}.

Figure~\ref{fig: imdp ex} shows an example of an IMDP adopted from \citep{hahn2019interval, monir2024lyapunov}, with the set of states $S = \{s, t, u\}$, the initial state $s$, and the set of actions $U = \{a,b\}$. The non-zero transition probability intervals are illustrated \citep[Fig. 5]{monir2024lyapunov}. Here, we consider two stationary policies $\pi \in \{1,2\}$, where each policy corresponds to a fixed action selection for all states and remains unchanged over time. The interval transition probability matrices are defined as
\[
\begin{medsize}
\begin{aligned}
\underbar{P}(1) &= \begin{bmatrix}
    0 & \frac{1}{3} & \frac{1}{10} \\ 
    0 & 1 & 0 \\ 
    0 & 0 & 1
\end{bmatrix}\!\!, \quad
\overbar{P}(1) = \begin{bmatrix}
    0 & \frac{2}{3} & 1 \\ 
    0 & 1 & 0 \\ 
    0 & 0 & 1
\end{bmatrix}\!\!, \quad
\underbar{P}(2) = \begin{bmatrix}
    0 & \frac{2}{5} & \frac{1}{10} \\ 
    0 & 1 & 0 \\ 
    0 & 0 & 1
\end{bmatrix}\!\!, \quad 
\overbar{P}(2) = \begin{bmatrix}
    0 & \frac{3}{5} & 1 \\ 
    0 & 1 & 0 \\ 
    0 & 0 & 1
\end{bmatrix}\!\!.
\end{aligned}
\end{medsize}
\]
\begin{figure}
    \centering
    \includegraphics[width=0.43\linewidth]{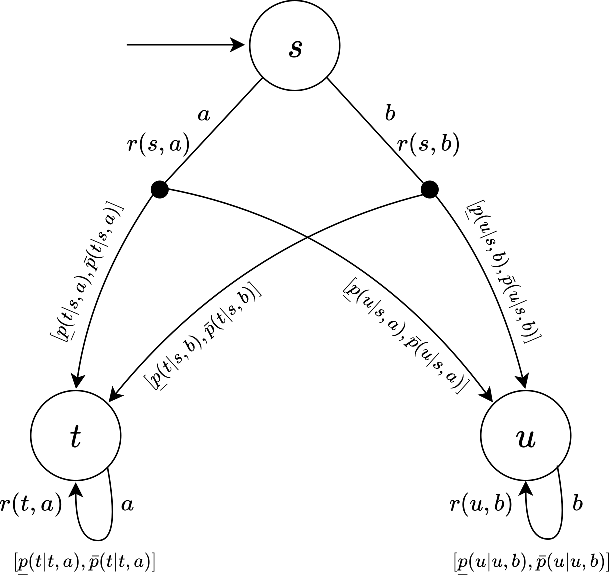}
    \caption{IMDP adopted from \citep{hahn2019interval, monir2024lyapunov}}
    \label{fig: imdp ex}
\end{figure}

The interval reward vectors are defined as $\underbar{R}(1) = \left[\frac{13}{10}; \frac{3}{10};
        \frac{1}{10}\right]$, $\overbar{R}(1) = \left[5; \frac{3}{10};
        \frac{1}{10}\right]$, $\underbar{R}(2) = \left[\frac{1}{2}; \frac{3}{10};
        \frac{1}{10}\right]$, and $\overbar{R}(2) = \left[\frac{8}{5}; \frac{3}{10};
        \frac{1}{10}\right]$. The discount factor $\gamma = 0.5$ and target $W_{tar}$ corresponding to $\lambda = {\left[\frac{9}{10};\frac{1}{10}\right]}$ and the initial values $W_0={\left[0;0;0\right]}$ are selected. 
        Note that we focus on IMDPs, where transition probabilities are represented as intervals to account for uncertainty. Each policy gives a specific action for each state, but the actual transition probabilities can vary within the specified bounds. Importantly, the upper-bound matrices do not need to be valid probability distributions and may have rows that sum to more than 1. During execution, a valid probability distribution (with rows summing to 1) is selected for each state-action pair within these bounds.

        The number of samples taken to solve scenario optimization problems are \(N = 700\) and \(M = 250\), with the calculated confidence being \(\beta = 0.05\). The results obtained via Algorithm \ref{alg} are illustrated in Figure~\ref{fig: imdp ex sim}. 
\begin{figure}[t]
    \centering
    \begin{minipage}[b]{0.43\linewidth}
        \includegraphics[width=\linewidth]{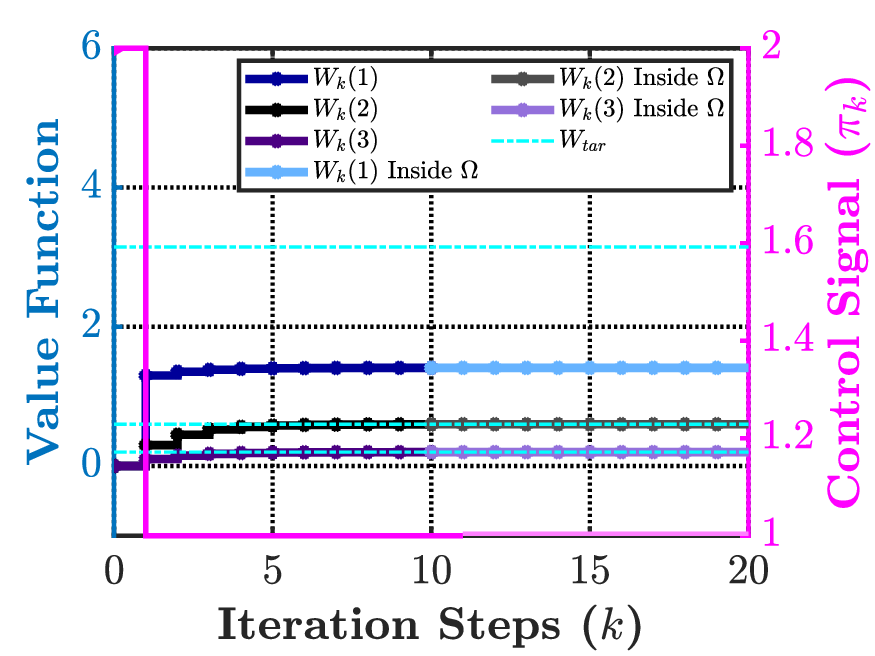}
        \centering (a)
    \end{minipage}
    \begin{minipage}[b]{0.43\linewidth}
        \includegraphics[width=\linewidth]{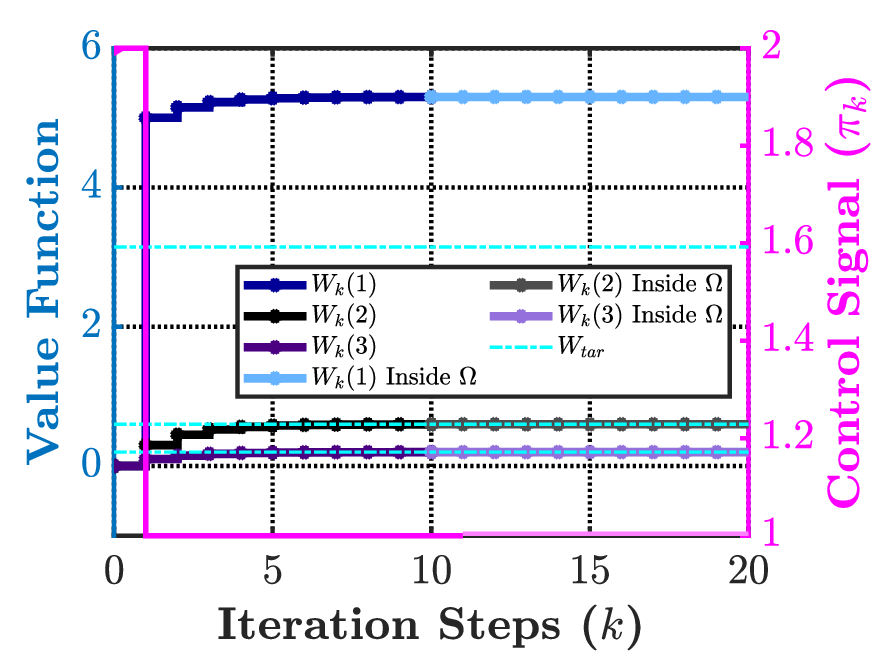}
        \centering (b)
    \end{minipage}
    \begin{minipage}[b]{0.43\linewidth}
        \includegraphics[width=\linewidth]{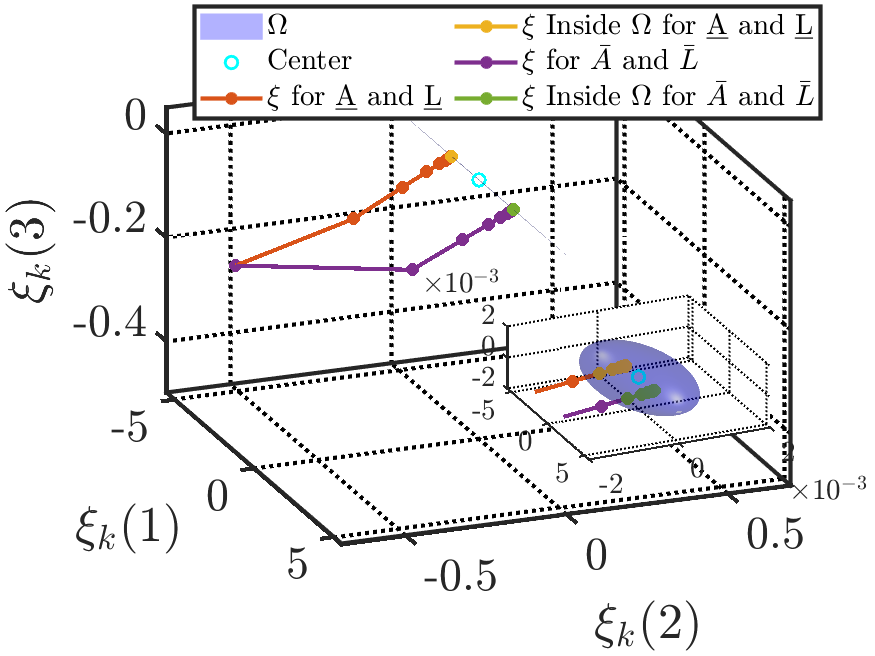}
        \centering (c)
    \end{minipage}
    \begin{minipage}[b]{0.43\linewidth}
        \includegraphics[width=\linewidth]{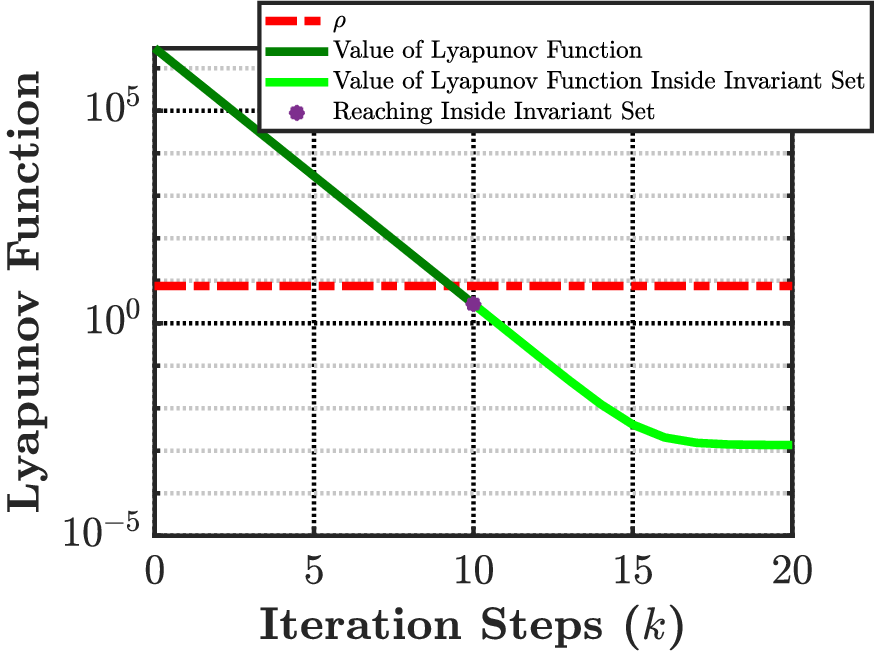}
        \centering (d)
    \end{minipage}
    \caption{The value function evolution is plotted for IMDP, considering $\underbar{A}_\pi$ and $\underbar{L}_\pi$ in (a), and $\overbar{A}_\pi$ and $\overbar{L}_\pi$ in (b). The switching law is indicated in pink in both (a) and (b). The invariant set of attraction with error trajectories for both lower- and upper-bound matrices is shown in (c). The Lyapunov function in (d) shows that the trajectories approach the invariant set of attraction and remain within the set.}
    \label{fig: imdp ex sim}
\end{figure}

To demonstrate our algorithm's applicability to MOIMDPs, we enhance the existing IMDP example by incorporating a second objective, resulting in an IMDP with two objectives. Our new objective is defined via the interval reward vectors $\underbar{R}(1) = \left[\frac{13}{30}; \frac{1}{5}; \frac{2}{5}\right]$, $\overbar{R}(1) = \left[5; \frac{1}{5}; \frac{2}{5}\right]$, $\underbar{R}(2) = \left[\frac{13}{16}; \frac{1}{5}; \frac{2}{5}\right]$, and $\overbar{R}(2) = \left[\frac{19}{12}; \frac{1}{5}; \frac{2}{5}\right]$, as well as a discount factor of $ \gamma =0.7$. Target $W_{tar}$ corresponding to $\lambda = {\left[\frac{9}{10};\frac{1}{10}\right]}$ and the initial values $W_0^1 = W_0^2={\left[0;0;0\right]}$ are selected. The number of samples in the scenario optimization problems is \(N = 1500\) and \(M = 300\), which gives the confidence \(\beta = 0.05\). The results obtained via Algorithm~\ref{alg} are illustrated in Figure~\ref{fig: moimdp}. 
\begin{figure}[h]
    \centering
    \begin{minipage}[b]{0.43\linewidth}
        \includegraphics[width=\linewidth]{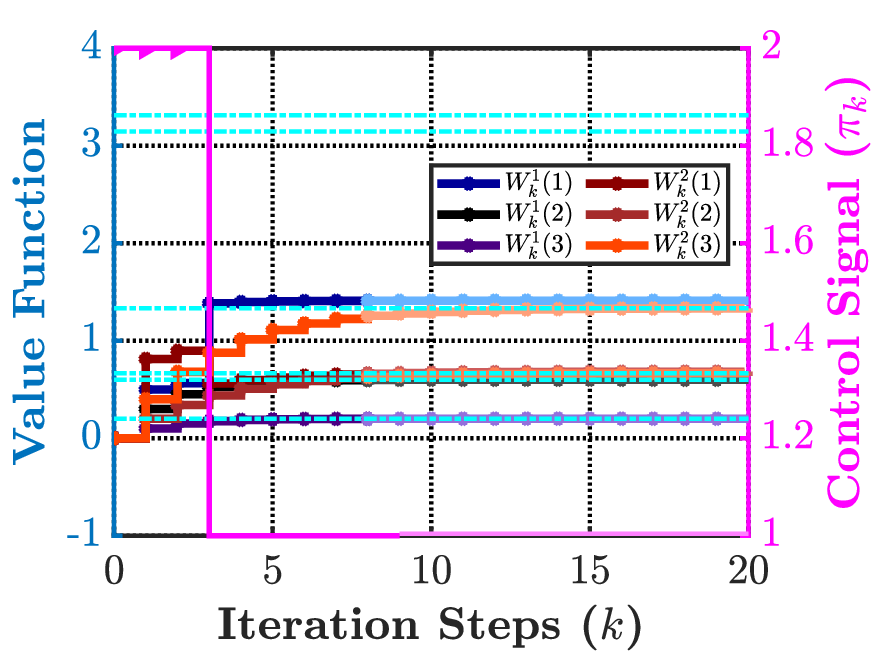}
        \centering (a)
    \end{minipage}
    \begin{minipage}[b]{0.43\linewidth}
        \includegraphics[width=\linewidth]{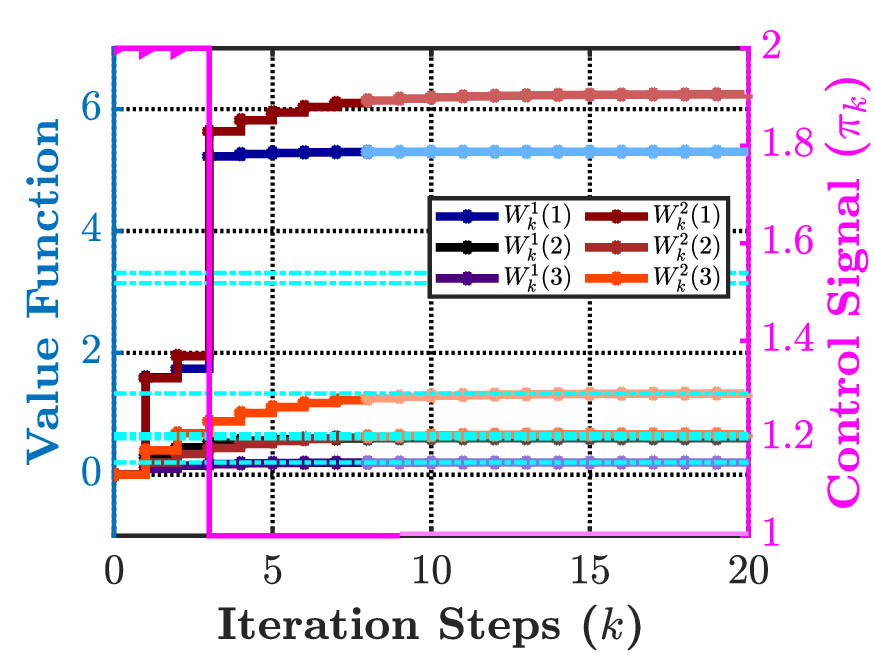}
        \centering (b)
    \end{minipage}
    \hfill
    \begin{minipage}[b]{0.43\linewidth}
        \includegraphics[width=\linewidth]{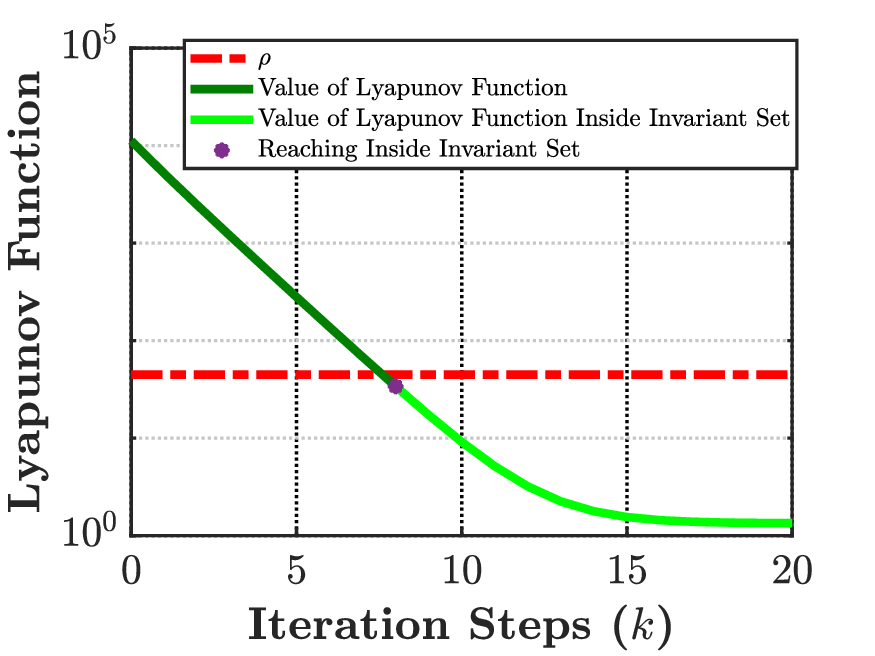}
        \centering (c)
    \end{minipage}
    \caption{The evolution of the value function for MOIMDP considering $\underbar{A}_\pi$ and $\underbar{L}_\pi$ in (a), and $\overbar{A}_\pi$ and $\overbar{L}_\pi$ in (b). The switching law is indicated in pink in both (a) and (b). The Lyapunov function in (c) shows that the trajectories approach the invariant set of attraction and remain within the set.}
    \label{fig: moimdp}
\end{figure}

\subsection{Power Electronic Systems}
The class of switched affine systems encompasses key features of various applications, including power electronic converters \citep{deaecto2010switched, albea2015hybrid, mojallizadeh2018hybrid}. Here, we applied Algorithm \ref{alg} to Single-Inductor Multiple-Output (SIMO) converter to show the effectiveness of our algorithm. The SIMO boost converter consists of three switches. Depending on the switch positions, the converter can be modeled as a switched affine system. We adopted the model and parameters from \citep{mojallizadeh2018hybrid} and discretized it using the Euler discretization method with an appropriate sampling step time of $\Delta t = 0.1$ milliseconds \citep{butcher2016numerical}. Based on model \eqref{eq. dt SAS}, the state vector is $x_k = [i^l_k;v^{o1}_k;v^{o2}_k]$, where $i^l$ is inductor current; $v^{o1}_k$ and $v^{o2}_k$ are output voltages. The system switching laws are defined as  $\pi \in \mathbb{N}_{\mathcal{P}}$ with $\mathcal{P}=3$. System matrices are
\[
\begin{medsize}
\begin{aligned}
        A_1\! =\! \begin{bmatrix}
            0 & 0 & 0 \\ 
            0 & 1-\frac{-\Delta t}{\mathcal{R}_1\mathcal{C}_1} & 0 \\
            0 & 0  & 1-\frac{-\Delta t}{\mathcal{R}_2\mathcal{C}_2}
        \end{bmatrix}\!\!,   A_2\! &=\! \begin{bmatrix}
            0 & -\frac{1}{\mathcal{L}} & 0 \\ 
            -\frac{1}{C_1} & 1-\frac{-\Delta t}{\mathcal{R}_1\mathcal{C}_1} & 0 \\
            0 & 0  & 1-\frac{-\Delta t}{\mathcal{R}_2\mathcal{C}_2}
        \end{bmatrix}\!\!,
    A_3 \!=\! \begin{bmatrix}
            0 &  0 & -\frac{1}{\mathcal{L}} \\ 
            0 & 1-\frac{-\Delta t}{\mathcal{R}_1\mathcal{C}_1} & 0 \\
            -\frac{1}{\mathcal{C}_1} & 0  & 1-\frac{-\Delta t}{\mathcal{R}_2\mathcal{C}_2}
            \end{bmatrix}\!\!,
\end{aligned}
\end{medsize}
\]
and $B_1 = B_2 = B_3 = \begin{bmatrix}
                \frac{v_i}{\mathcal{L}}; 0 ; 0
            \end{bmatrix}$,
where $v_i \in [20,30] $ \emph{V} is input voltage; $\mathcal{R}_1 \in [39,56]$ \emph{$\Omega$} and $\mathcal{R}_2 \in [47,68]$ \emph{$\Omega$} are resistive loads; $\mathcal{L} = 5$ \emph{$mH$} denotes the inductance; and $\mathcal{C}_1 = \mathcal{C}_2=470$ \emph{$\mu F$} are the capacitance of capacitors. The parameters are based on the laboratory prototype used in \citep{mojallizadeh2018hybrid}.  
    Since $\mathcal{R}_1$, $\mathcal{R}_2$, and $v_i$ are uncertain, it is straightforward to determine $[\underbar{A}_\pi, \overbar{A}_\pi]$ and $[\underbar{B}_\pi, \overbar{B}_\pi]$ based on the upper and lower bounds of the uncertain parameters. The operating point $x_{e}$ 
    corresponding to $\lambda = {[0 ;0.5;0.5]}$ and the initial values $x_0={\left[0;0;0\right]}$ are selected. 
    The number of samples for scenario optimization problems are \(N = 900\) and \(M = 200\), with confidence \(\beta = 0.06\). The results obtained via Algorithm~\ref{alg} are illustrated in Figure~\ref{fig: simo ex}. 
    \begin{figure}[t]
    \centering
    \begin{minipage}[b]{0.43\linewidth}
        \includegraphics[width=\linewidth]{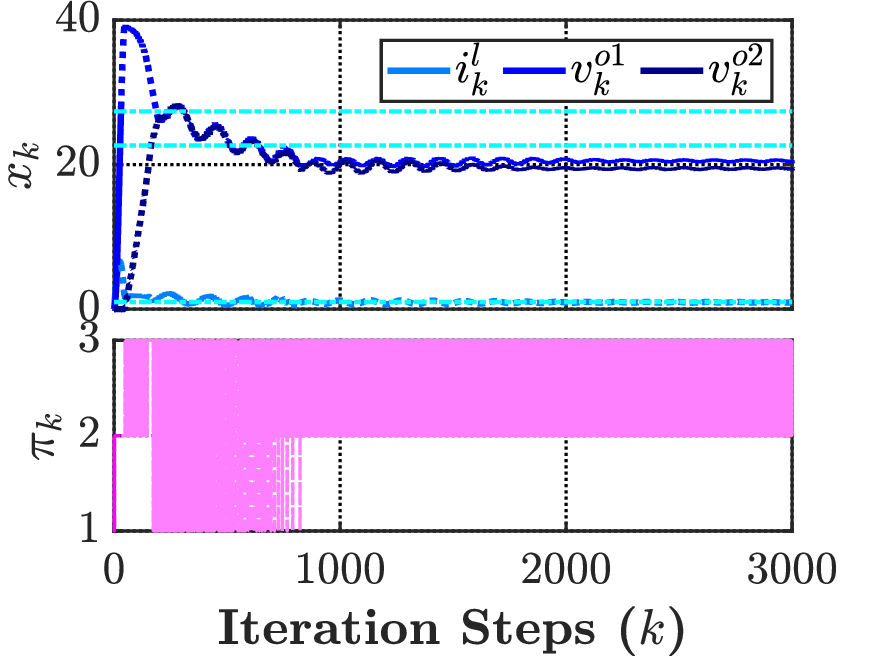}
        \centering (a)
    \end{minipage}
    \begin{minipage}[b]{0.43\linewidth}
        \includegraphics[width=\linewidth]{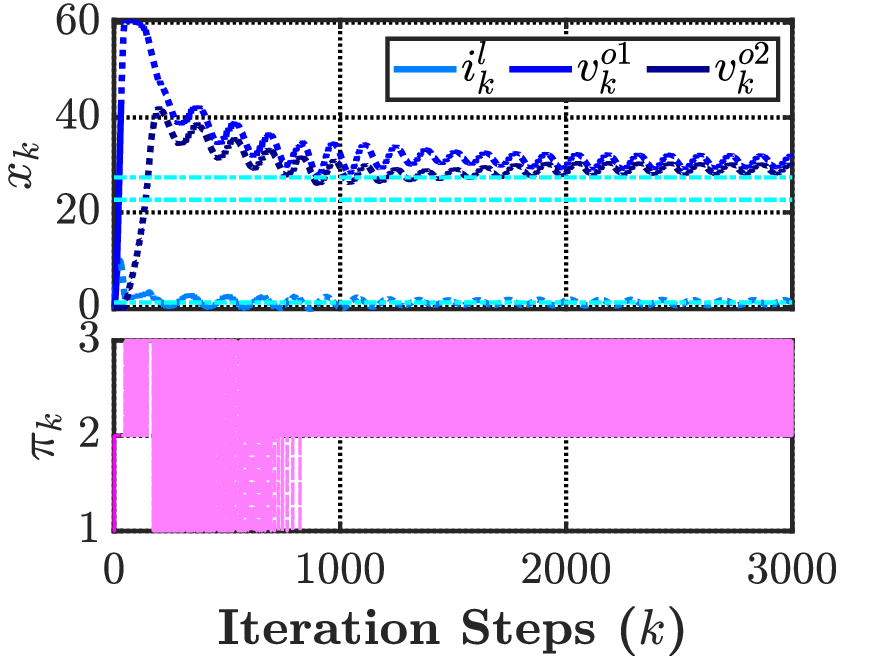}
        \centering (b)
    \end{minipage}
    \begin{minipage}[b]{0.43\linewidth}
        \includegraphics[width=\linewidth]{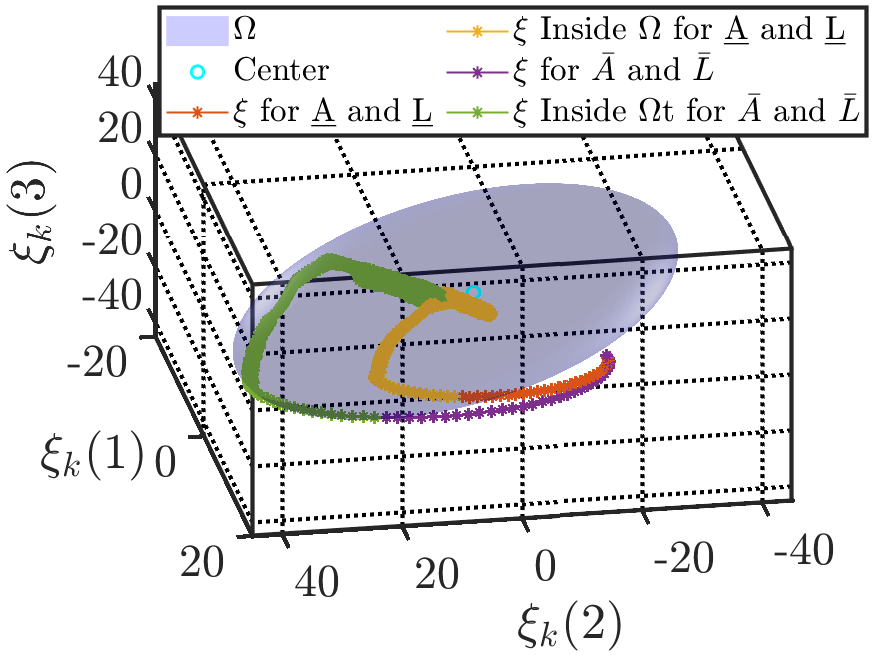}
        \centering (c)
    \end{minipage}
    \begin{minipage}[b]{0.43\linewidth}
        \includegraphics[width=\linewidth]{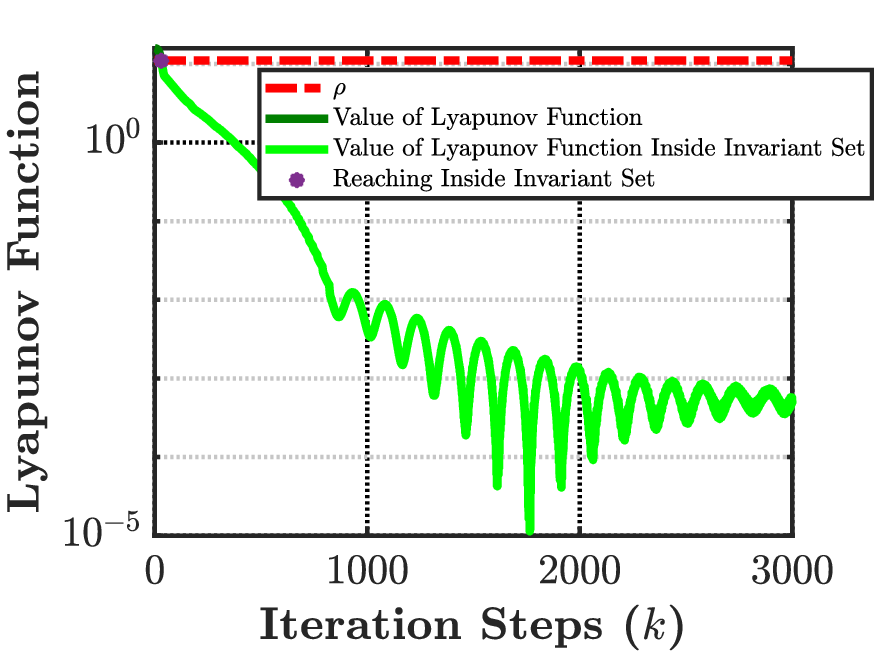}
        \centering (d)
    \end{minipage}
    \caption{The evolution of $x_k$ for lower- and upper-bound matrices is illustrated in (a) and  (b). The desired operating point $x_e$ is indicated by light blue in both (a) and (b). The switching law is shown in pink in both (a) and (b), in steady state, it is switching between $\pi = 2$ and $\pi =3$. The invariant set of attraction with error trajectories for both lower- and upper-bound matrices is expressed in (d). The Lyapunov function in (c) shows that the trajectories approach the invariant set of attraction and remain within the set.}
    \label{fig: simo ex}
\end{figure}

\subsection{Comparison with an alternative approach}
\label{app ex num}

We provide an example, adapted from \citep{albea2020robust}, to compare the tightness of the invariant set in our approach with the previous method proposed by \citet{albea2020robust}.

Consider an uncertain system described by \eqref{eq. dt SAS} that consists of three functioning modes
\begin{align*}
    &A_1 = \begin{bmatrix}
        0 & 0.15+\delta \\
        -0.35 & -1
    \end{bmatrix}, \quad
    A_2 = \begin{bmatrix}
        0.24 & 0.15+\delta \\
        -2.35 & -1
    \end{bmatrix}, \quad
    A_3 = \begin{bmatrix}
        -0.24 & 0.15+\delta \\
        -2.35 & -0.5
    \end{bmatrix}, \\[1em]
    &B_1 = \begin{bmatrix}
        1 \\
        0.35
    \end{bmatrix}, \quad 
    B_2 = \begin{bmatrix}
        -1 \\
        -0.35
    \end{bmatrix}, \quad
    B_3 = \begin{bmatrix}
        0.05 \\
        1.5
    \end{bmatrix},
\end{align*}
where $\delta$ represents an unknown, potentially time-varying perturbation to the system. It is assumed to be bounded by a known parameter $\bar{\delta} \in [0, 0.2]$. The desired equilibrium point is given by \( x_e = \begin{bmatrix} 0.1 ; 0.2 \end{bmatrix} \) with the parameter \( \lambda = [ 0.36 ; 0.3 ; 0.34] \), where the nominal value is considered to be \( \delta = 0 \). We applied Algorithm \ref{alg} to this example and considering $\delta = 0.2$ to examine the worst-case uncertainty. The number of samples for the scenario optimization are \(N = 300\) and \(M = 150\), with confidence parameter \(\beta = 0.04\). The results are shown in Figure \ref{fig: albae num}. 
The invariant set designed in this paper is more precise
compared to the invariant set proposed in \citep{albea2020robust}. This reduces chattering effects and increases accuracy. \label{Ex albae}
\begin{figure}[t]
    \centering
    \begin{minipage}[b]{0.43\linewidth}
        \includegraphics[width=\linewidth]{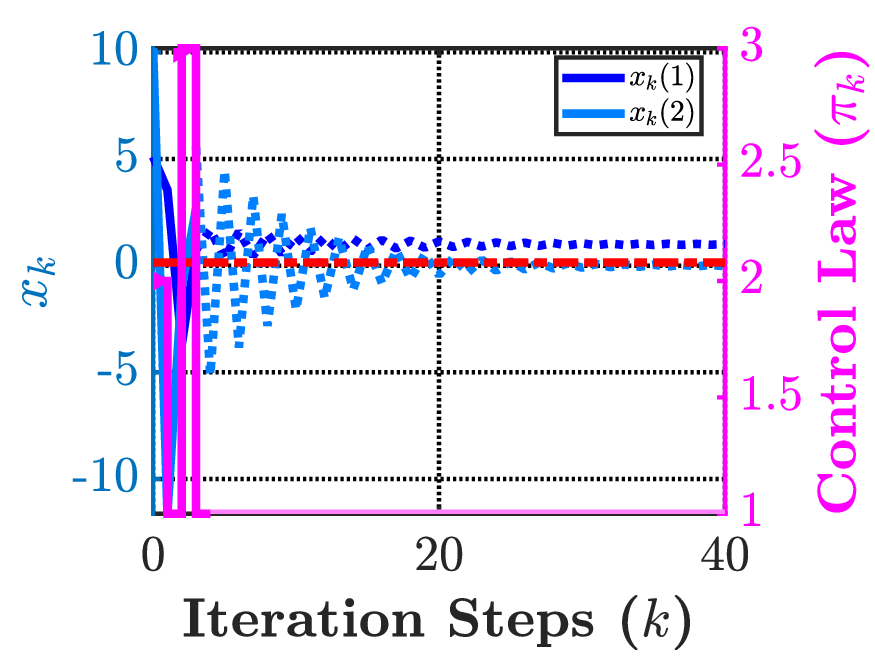}
        \centering (a)
    \end{minipage}
    \begin{minipage}[b]{0.43\linewidth}
        \includegraphics[width=\linewidth]{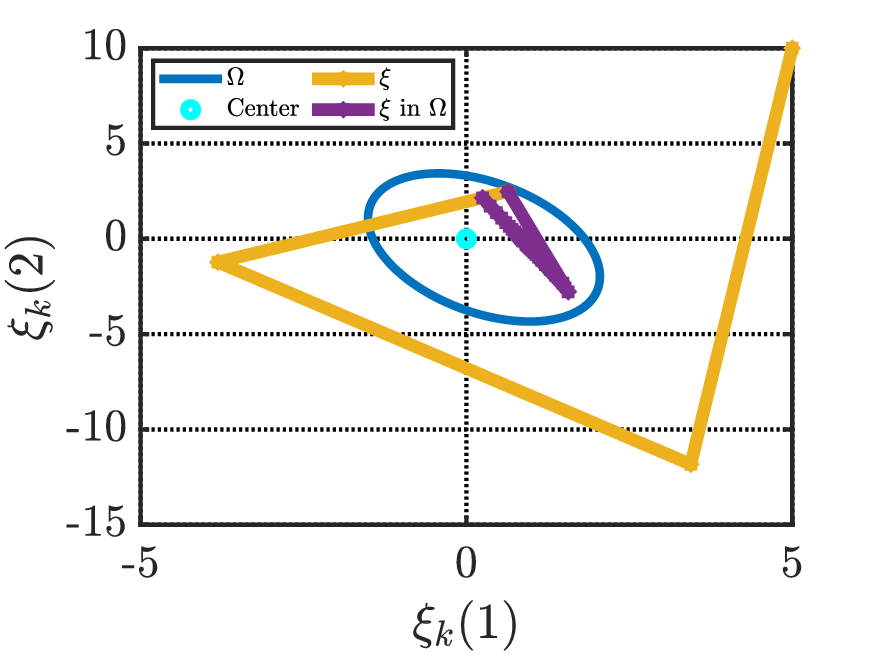}
        \centering (b)
    \end{minipage}
    \hfill
    \begin{minipage}[b]{0.43\linewidth}
        \includegraphics[width=\linewidth]{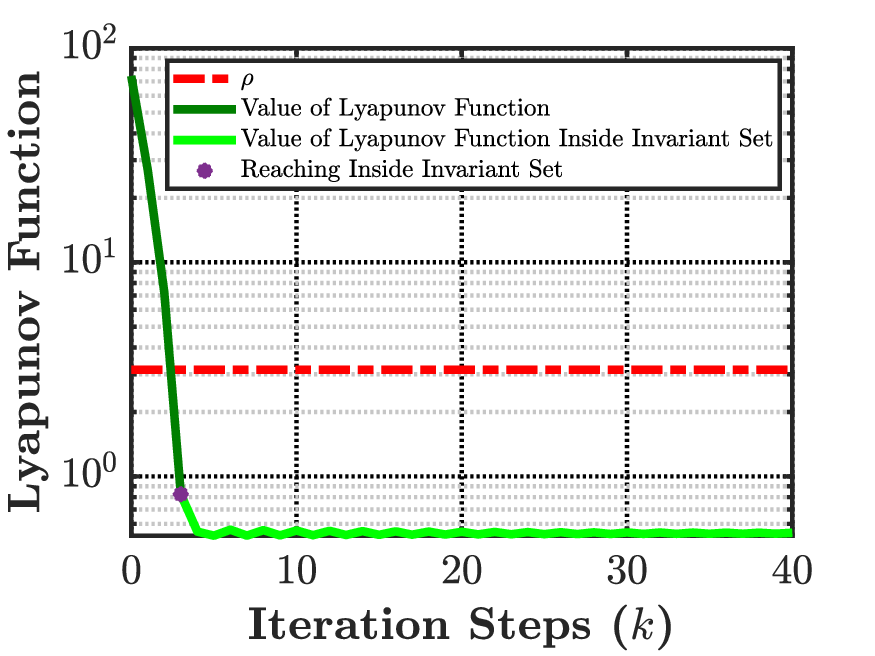}
        \centering (c)
    \end{minipage}
    \caption{The evolution of $x_k$ for $\delta = 0.2$ is illustrated in (a). The desired operating point $x_e$ is indicated by red, and the switching law is shown in pink in (a). The invariant set of attraction with error trajectories for $\delta = 0.2$ is expressed in (b). The Lyapunov function in (c) shows that the trajectories approach the invariant set of attraction and remain within the set.}
    \label{fig: albae num}
\end{figure}

\section{Conclusion} \label{sec conc}
This paper introduced a novel approach for designing switching control laws in uncertain switched affine systems through data-driven scenario programming. The effectiveness of the proposed approach was demonstrated through case studies from uncertain Markov processes and power electronic systems. Future work will focus on making the scenario program robust to unseen data, improving the scalability of the approach to handle large-scale MDPs and DC microgrids, and considering time-varying Lyapunov functions to address tracking problems in microgrids.

\newpage

\acks{The research of N. Monir is supported by the EPSRC EP/W524700/1 and Newcastle University Global Scholarship. The research of S. Soudjani is supported by the following grants: EIC 101070802 and ERC 101089047.}
\bibliography{biblio}

\begin{thebibliography}{41}
\providecommand{\natexlab}[1]{#1}
\providecommand{\url}[1]{\texttt{#1}}
\expandafter\ifx\csname urlstyle\endcsname\relax
  \providecommand{\doi}[1]{doi: #1}\else
  \providecommand{\doi}{doi: \begingroup \urlstyle{rm}\Url}\fi

\bibitem[Albea et~al.(2015)Albea, Garcia, and Zaccarian]{albea2015hybrid}
C.~Albea, G.~Garcia, and L.~Zaccarian.
\newblock Hybrid dynamic modeling and control of switched affine systems: application to dc-dc converters.
\newblock In \emph{2015 54th IEEE Conference on Decision and Control (CDC)}, pages 2264--2269. IEEE, 2015.

\bibitem[Albea~Sanchez et~al.(2020)Albea~Sanchez, Ventosa-Cutillas, Seuret, and Gordillo]{albea2020robust}
C.~Albea~Sanchez, A.~Ventosa-Cutillas, A.~Seuret, and F.~Gordillo.
\newblock Robust switching control design for uncertain discrete-time switched affine systems.
\newblock \emph{International Journal of Robust and Nonlinear Control}, 30\penalty0 (17):\penalty0 7089--7102, 2020.

\bibitem[Beneux et~al.(2018)Beneux, Astolfi, Riedinger, Daafouz, and Grimaud]{beneux2018integral}
G.~Beneux, D.~Astolfi, P.~Riedinger, J.~Daafouz, and L.~Grimaud.
\newblock Integral action for uncertain switched affine systems with application to dc/dc converters.
\newblock In \emph{2018 European Control Conference (ECC)}, pages 795--800. IEEE, 2018.

\bibitem[Berberich et~al.(2020)Berberich, K{\"o}hler, M{\"u}ller, and Allg{\"o}wer]{berberich2020data}
J.~Berberich, J.~K{\"o}hler, M.A. M{\"u}ller, and F.~Allg{\"o}wer.
\newblock Data-driven model predictive control with stability and robustness guarantees.
\newblock \emph{IEEE Transactions on Automatic Control}, 66\penalty0 (4):\penalty0 1702--1717, 2020.

\bibitem[Boyd and Vandenberghe(2004)]{boyd2004convex}
S.~Boyd and L.~Vandenberghe.
\newblock \emph{Convex optimization}.
\newblock Cambridge university press, 2004.

\bibitem[Butcher(2016)]{butcher2016numerical}
J.C. Butcher.
\newblock \emph{Numerical methods for ordinary differential equations}.
\newblock John Wiley \& Sons, 2016.

\bibitem[Campi and Garatti(2018{\natexlab{a}})]{campi2018introduction}
M.C. Campi and S.~Garatti.
\newblock \emph{Introduction to the scenario approach}.
\newblock SIAM, 2018{\natexlab{a}}.

\bibitem[Campi and Garatti(2018{\natexlab{b}})]{campi2018wait}
M.C. Campi and S.~Garatti.
\newblock Wait-and-judge scenario optimization.
\newblock \emph{Mathematical Programming}, 167:\penalty0 155--189, 2018{\natexlab{b}}.

\bibitem[Campi et~al.(2009)Campi, Garatti, and Prandini]{campi2009scenario}
M.C. Campi, S.~Garatti, and M.~Prandini.
\newblock The scenario approach for systems and control design.
\newblock \emph{Annual Reviews in Control}, 33\penalty0 (2):\penalty0 149--157, 2009.

\bibitem[Campi et~al.(2021{\natexlab{a}})Campi, Car{\`e}, and Garatti]{campi2021scenario}
M.C. Campi, A.~Car{\`e}, and S.~Garatti.
\newblock The scenario approach: A tool at the service of data-driven decision making.
\newblock \emph{Annual Reviews in Control}, 52:\penalty0 1--17, 2021{\natexlab{a}}.

\bibitem[Campi et~al.(2021{\natexlab{b}})Campi, Carè, and Garatti]{CAMPI20211}
M.C. Campi, A.~Carè, and S.~Garatti.
\newblock The scenario approach: A tool at the service of data-driven decision making.
\newblock \emph{Annual Reviews in Control}, 52:\penalty0 1--17, 2021{\natexlab{b}}.
\newblock ISSN 1367-5788.

\bibitem[De~Persis and Tesi(2019)]{de2019formulas}
C.~De~Persis and P.~Tesi.
\newblock Formulas for data-driven control: Stabilization, optimality, and robustness.
\newblock \emph{IEEE Transactions on Automatic Control}, 65\penalty0 (3):\penalty0 909--924, 2019.

\bibitem[Deaecto and Geromel(2016)]{deaecto2016stability}
G.S. Deaecto and J.C. Geromel.
\newblock Stability analysis and control design of discrete-time switched affine systems.
\newblock \emph{IEEE Transactions on Automatic Control}, 62\penalty0 (8):\penalty0 4058--4065, 2016.

\bibitem[Deaecto et~al.(2010)Deaecto, Geromel, Garcia, and Pomilio]{deaecto2010switched}
G.S. Deaecto, J.C. Geromel, F.S. Garcia, and J.A. Pomilio.
\newblock Switched affine systems control design with application to dc--dc converters.
\newblock \emph{IET control theory \& applications}, 4\penalty0 (7):\penalty0 1201--1210, 2010.

\bibitem[Della~Rossa et~al.(2021)Della~Rossa, Wang, Egidio, and Jungers]{della2021data}
M.~Della~Rossa, Z.~Wang, L.N. Egidio, and R.M. Jungers.
\newblock Data-driven stability analysis of switched affine systems.
\newblock In \emph{2021 60th IEEE Conference on Decision and Control (CDC)}, pages 3204--3209. IEEE, 2021.

\bibitem[Dietrich et~al.(2024)Dietrich, Devonport, and Arcak]{dietrich2024nonconvex}
E.~Dietrich, A.~Devonport, and M.~Arcak.
\newblock Nonconvex scenario optimization for data-driven reachability.
\newblock In \emph{6th Annual Learning for Dynamics \& Control Conference}, pages 514--527. PMLR, 2024.

\bibitem[Galambos(1977)]{galambos1977bonferroni}
J.~Galambos.
\newblock Bonferroni inequalities.
\newblock \emph{The Annals of Probability}, pages 577--581, 1977.

\bibitem[Garatti and Campi(2024)]{garatti2024non}
S.~Garatti and M.C. Campi.
\newblock Non-convex scenario optimization.
\newblock \emph{Mathematical Programming}, pages 1--52, 2024.

\bibitem[Grant and Boyd(2014)]{cvx}
M.~Grant and S.~Boyd.
\newblock {CVX}: Matlab software for disciplined convex programming, version 2.1.
\newblock \url{https://cvxr.com/cvx}, March 2014.

\bibitem[Hahn et~al.(2019)Hahn, Hashemi, Hermanns, Lahijanian, and Turrini]{hahn2019interval}
E.M. Hahn, V.~Hashemi, H.~Hermanns, M.~Lahijanian, and A.~Turrini.
\newblock Interval {Markov} decision processes with multiple objectives: From robust strategies to {Pareto} curves.
\newblock \emph{ACM Transactions on Modeling and Computer Simulation (TOMACS)}, 29\penalty0 (4):\penalty0 1--31, 2019.

\bibitem[Iervolino et~al.(2023)Iervolino, Tipaldi, and Forootani]{iervolino2023lyapunov}
R.~Iervolino, M.~Tipaldi, and A.~Forootani.
\newblock A lyapunov-based version of the value iteration algorithm formulated as a discrete-time switched affine system.
\newblock \emph{International Journal of Control}, 96\penalty0 (3):\penalty0 577--592, 2023.

\bibitem[Kazemi et~al.(2024)Kazemi, Majumdar, Salamati, Soudjani, and Wooding]{kazemi2024data}
Milad Kazemi, Rupak Majumdar, Mahmoud Salamati, Sadegh Soudjani, and Ben Wooding.
\newblock Data-driven abstraction-based control synthesis.
\newblock \emph{Nonlinear Analysis: Hybrid Systems}, 52:\penalty0 101467, 2024.

\bibitem[Kordabad et~al.(2025)Kordabad, Vlahakis, Lindemann, Gros, Dimarogonas, and Soudjani]{kordabad2025data}
Arash~Bahari Kordabad, Eleftherios~E Vlahakis, Lars Lindemann, Sebastien Gros, Dimos~V Dimarogonas, and Sadegh Soudjani.
\newblock Data-driven distributionally robust control for interacting agents under logical constraints.
\newblock \emph{arXiv preprint arXiv:2503.09816}, 2025.

\bibitem[Laino et~al.(2025)Laino, Wooding, Soudjani, and Davenport]{laino2025logic}
Anna~S Laino, Ben Wooding, Sadegh Soudjani, and Russell~J Davenport.
\newblock A logic-based resilience metric for water resource recovery facilities.
\newblock \emph{Environmental Science: Water Research \& Technology}, 11\penalty0 (2):\penalty0 377--392, 2025.

\bibitem[Makdesi et~al.(2021)Makdesi, Girard, and Fribourg]{makdesi2021data}
Anas Makdesi, Antoine Girard, and Laurent Fribourg.
\newblock Data-driven abstraction of monotone systems.
\newblock In \emph{Learning for Dynamics and Control}, pages 803--814. PMLR, 2021.

\bibitem[Mojallizadeh and Badamchizadeh(2018)]{mojallizadeh2018hybrid}
M.R. Mojallizadeh and M.A. Badamchizadeh.
\newblock Hybrid control of single-inductor multiple-output converters.
\newblock \emph{IEEE Transactions on Industrial Electronics}, 66\penalty0 (1):\penalty0 451--458, 2018.

\bibitem[Monir et~al.(2024)Monir, Sch{\"o}n, and Soudjani]{monir2024lyapunov}
N.~Monir, O.~Sch{\"o}n, and S.~Soudjani.
\newblock Lyapunov-based policy synthesis for multi-objective interval mdps.
\newblock \emph{IFAC-PapersOnLine}, 58\penalty0 (11):\penalty0 99--106, 2024.

\bibitem[Monir et~al.(2025)Monir, Sadabadi, and Soudjani]{monir2025logic}
Negar Monir, Mahdieh~S Sadabadi, and Sadegh Soudjani.
\newblock Logic-based resilience computation of power systems against frequency requirements.
\newblock \emph{European Control Conference}, 2025.

\bibitem[Nazeri et~al.(2025)Nazeri, Badings, Soudjani, and Abate]{nazeri2025data}
Mahdi Nazeri, Thom Badings, Sadegh Soudjani, and Alessandro Abate.
\newblock Data-driven yet formal policy synthesis for stochastic nonlinear dynamical systems.
\newblock \emph{Learning for Dynamics and Control}, 2025.

\bibitem[Salamati and Zamani(2022)]{salamati2022data}
A.~Salamati and M.~Zamani.
\newblock Data-driven safety verification of stochastic systems via barrier certificates: A wait-and-judge approach.
\newblock In \emph{Learning for Dynamics and Control Conference}, pages 441--452. PMLR, 2022.

\bibitem[Salamati et~al.(2024)Salamati, Lavaei, Soudjani, and Zamani]{salamati2024data}
Ali Salamati, Abolfazl Lavaei, Sadegh Soudjani, and Majid Zamani.
\newblock Data-driven verification and synthesis of stochastic systems via barrier certificates.
\newblock \emph{Automatica}, 159:\penalty0 111323, 2024.

\bibitem[Saoud et~al.(2024)Saoud, Jagtap, and Soudjani]{saoud2024temporal}
Adnane Saoud, Pushpak Jagtap, and Sadegh Soudjani.
\newblock Temporal logic resilience for dynamical systems.
\newblock \emph{arXiv preprint arXiv:2404.19223}, 2024.

\bibitem[Sch{\"o}n et~al.(2024{\natexlab{a}})Sch{\"o}n, van Huijgevoort, Haesaert, and Soudjani]{schon2024bayesian}
Oliver Sch{\"o}n, Birgit van Huijgevoort, Sofie Haesaert, and Sadegh Soudjani.
\newblock Bayesian formal synthesis of unknown systems via robust simulation relations.
\newblock \emph{IEEE Transactions on Automatic Control}, 2024{\natexlab{a}}.

\bibitem[Sch{\"o}n et~al.(2024{\natexlab{b}})Sch{\"o}n, Zhong, and Soudjani]{schon2024data}
Oliver Sch{\"o}n, Zhengang Zhong, and Sadegh Soudjani.
\newblock Data-driven distributionally robust safety verification using barrier certificates and conditional mean embeddings.
\newblock In \emph{2024 American Control Conference (ACC)}, pages 3417--3423. IEEE, 2024{\natexlab{b}}.

\bibitem[Seatzu et~al.(2006)Seatzu, Corona, Giua, and Bemporad]{seatzu2006optimal}
C.~Seatzu, D.~Corona, A.~Giua, and A.~Bemporad.
\newblock Optimal control of continuous-time switched affine systems.
\newblock \emph{IEEE transactions on automatic control}, 51\penalty0 (5):\penalty0 726--741, 2006.

\bibitem[Seuret et~al.(2023)Seuret, Albea, and Gordillo]{seuret2023practical}
A.~Seuret, C.~Albea, and F.~Gordillo.
\newblock Practical stabilization of switched affine systems: Model and data-driven conditions.
\newblock \emph{IEEE Control Systems Letters}, 7:\penalty0 1628--1633, 2023.

\bibitem[Skovbekk et~al.(2025)Skovbekk, Laurenti, Frew, and Lahijanian]{skovbekk2025formal}
John Skovbekk, Luca Laurenti, Eric Frew, and Morteza Lahijanian.
\newblock Formal verification of unknown dynamical systems via {G}aussian process regression.
\newblock \emph{IEEE Transactions on Automatic Control}, 2025.

\bibitem[Van~Waarde et~al.(2020)Van~Waarde, Eising, Trentelman, and Camlibel]{van2020data}
H.J. Van~Waarde, J.~Eising, H.L. Trentelman, and M.K. Camlibel.
\newblock Data informativity: A new perspective on data-driven analysis and control.
\newblock \emph{IEEE Transactions on Automatic Control}, 65\penalty0 (11):\penalty0 4753--4768, 2020.

\bibitem[Wang et~al.(2024)Wang, Berger, and Jungers]{wang2024data}
Zheming Wang, Guillaume~O Berger, and Rapha{\"e}l~M Jungers.
\newblock Data-driven control of unknown switched linear systems using scenario optimization.
\newblock \emph{IEEE Transactions on Automatic Control}, 2024.

\bibitem[Wu et~al.(2023)Wu, Zhang, Haesaert, Ma, and Sun]{wu2023risk}
Lin-Chi Wu, Zengjie Zhang, Sofie Haesaert, Zhiqiang Ma, and Zhiyong Sun.
\newblock Risk-aware reward shaping of reinforcement learning agents for autonomous driving.
\newblock In \emph{IECON 2023-49th Annual Conference of the IEEE Industrial Electronics Society}, pages 1--6. IEEE, 2023.

\bibitem[Zhang et~al.(2022)Zhang, Lou, and Wang]{zhang2022output}
Lingyu Zhang, Xuyang Lou, and Zhan Wang.
\newblock Output-based robust switching rule design for uncertain switched affine systems: Application to dc--dc converters.
\newblock \emph{IEEE Transactions on Circuits and Systems II: Express Briefs}, 69\penalty0 (11):\penalty0 4493--4497, 2022.

\end{thebibliography}

\end{document}